\documentclass[onecolumn,draftclsnofoot,12pt]{IEEEtran}

\usepackage{amsmath, amssymb,graphicx, url}

\graphicspath{{figs/}}
\usepackage{bbm}
\usepackage{amssymb}
\usepackage{amsmath,amsthm}
\usepackage{color}
\usepackage{amsfonts}
\usepackage{graphicx}
\usepackage{verbatim}
\usepackage{epstopdf}
\usepackage{cite}
\usepackage{tabulary}
\usepackage{mathtools}
\usepackage{subcaption}

\usepackage{enumitem}
\usepackage{pgffor}

\newcommand{\figwidth}{4in}

\newcommand{\dvspace}[1]{}

\newcommand{\SigRadiusM}[1]{b^k_\mathrm{D}}

\newcommand{\dd}{\mathrm{d}}

\newcommand{\ie}{{\em i.e. } }

\newcommand{\expS}[1]{\exp{\left(#1\right)}}

\newcommand{\ths}{\text{th}}

\def\home{\hbox{\kern3pt \vbox to13pt{}%
		\pdfliteral{q 0 0 m 0 5 l 5 10 l 10 5 l 10 0 l 7 0 l 7 5 l 3 5 l 3 0 l f
			1 j 1 J -2 5 m 5 12 l 12 5 l S Q }%
		\kern 13pt}}

\newtheorem{theorem}{Theorem}

\newtheorem{corollary}{Corollary}

\usepackage{algorithm,algorithmic}
\usepackage{bbm}
\usepackage{caption}
\title{\huge Modeling and Analysis of Heterogeneous Traffic Networks with Anarchists and Socialist Traffic}
\author{Abhishek K Gupta and Adrish Banerjee\thanks{The authors are with the department of Electrical Engineering at IIT Kanpur, India 208016. Email:gkrabhi@iitk.ac.in.}}

\newcommand{\pathall}{\mathcal{P}}
\newcommand{\nodes}{\mathsf{V}}
\newcommand{\edges}{\mathsf{E}}
\newcommand{\usertypeset}{\mathcal{M}}
\newcommand{\socialist}{\mathsf{s}}
\newcommand{\anarchist}{\mathsf{a}}
\newcommand{\pfsocialist}{\mathsf{f}}
\newcommand{\flow}{f}
\newcommand{\pathflow}{f}

\begin{document}
\maketitle

\begin{abstract}
In this paper, we consider a heterogeneous traffic network with  multiple users classes which differ considerably in their path selection objective.  In particular, we consider two classes of users: ones who seek to minimize social cost (socialists) and the ones with typical greedy objective (anarchists) which leads to a heterogeneous game termed as HetGame.  The paper proposes an analytical framework to derive optimal/equilibrium flow in such a heterogeneous game along with a method for the same. The paper considers multiple examples of different networks to derive the optimal traffic assignment.  We introduce two metrics: price of $\alpha$ anarchy and price of good behavior to evaluate the impact of anarchists and implications of central directives. Finally,  the proposed algorithm is implemented for a real traffic network to derive insights.
\end{abstract}

\section{Introduction}
Traffic congestion is an important problem for planning of any urban city, owing to increasing traffic every day. Congestion games are an important framework for studying real-life traffic patterns, both from the road-network perspective and communication network perspective. Owing to a large number of possible routes to reach  their destination, it  may be  difficult for the users to choose the optimal path. Moreover, current traffic systems are complex and heterogeneous. One interesting source of heterogeneity is real-time traffic information. For example, consider a simple network in which some users follow route recommendations from an in-vehicle navigation system while others follow their own understanding. Such a scenario results in a so-called HetGame, a congestion game among heterogeneous users with different path selection objective.  The dynamics and equilibria of HetGames can provide insights in network planning and centralized user routing.

\textit{Related Work:}
Optimality of a path itself depends on the user's perspective. Most literature on optimal traffic assignment is divided into two main directions owing to the following two objectives: 
(1) The individual perspective to decide a route by choosing the least costly available path or (2) the socialist behavior to chose paths which minimizes the average cost of everyone. The first strategy may lead to an equilibrium which is the Nash equilibrium (NE) of corresponding congestion game and can be formed as an optimization problem using Wardrop's equilibrium conditions \cite{Wardrop1952}. However, this strategy is not optimal from the perspective of the city government or the city planner (e.g. department of transportation, or city government) \cite{LEBLANC1975309}. The city planner would want the second strategy so that the total travel cost in the city to be minimized, which is popularly also known as the social optimal situation. Along with analytical studies, past literature has proposed many algorithms and numerical methods to solve the problem in an iterative manner. The problem of traffic assignment  can be written as a general convex optimization which can be solved with standard optimization techniques. One promising technique is to use the Frank-Wolfe Algorithm to determine the optimal flows \cite{LEBLANC1975309}. Along with networks with homogeneous users, there has been some work in studying network with some heterogeneity of users. For example, in \cite{Cole2003}, a traffic network with users with varying trade-off preference between minimizing the latency it experiences and minimizing the money it is required to pay was studied.  Scenarios where a user fraction can cooperate  (known as  Stackelberg equilibrium) was studied in \cite{Swamy2012}. In \cite{Gupta2018}, a traffic network where random users are coexistent along with the regular traffic was analyzed to derive the optimal flow. In \cite{beckmann1956studies}, it was shown that tolls can be used to derive incentives to make nash equilibrium and social optimum the same.  Prior literature studies congestion games in which users have different utility function parameters\cite{Lo2005,Jiang2011}.  These parameters can model varying sensitivity to money, risk, fuel consumption. \cite{Nikolova2017} proposed deterministic strategies for central planner in order to provide desired  flows, including by randomly routing players after giving them specific guarantees about their costs. Along with networks with social and selfish objectives, there can be networks where both such users co-exist. Some of the users are ready to obey central directives and some of them are purely selfish. All the mentioned work have not studied networks having users of such heterogeneous nature which is the main focus of the work.

\textit{Contributions:}
 In this paper, we consider a heterogeneous traffic network with  multiple users classes which differ considerably in their path selection objective.  In particular, we consider two classes of users: ones who seek to minimize social cost (socialists) and the ones with typical greedy objective (anarchists).  This work examines non-atomic congestion games with these two user classes. We develop a framework to derive optimal/equilibrium flow in such a heterogeneous game and propose an algorithm for the same. We also propose two metrics: price of $\alpha$ anarchy and price of good behavior to evaluate the impact of anarchists and implications of central directives. We consider multiple networks to derive the optimal traffic assignment. With the help of analysis, we derive insights about these systems to  help formulate central directives which can  make the social optimal solution to be equal to the equilibrium. We also evaluate the proposed algorithm for a real traffic network.

\section{System Model}

In this paper, we consider a traffic network with heterogeneous users (commuters/packets) termed as heterogeneous traffic network (HetTN). The traffic network $N$ is modelled as a graph $\mathcal{G}$ with nodes $\nodes$ and edges $\edges$. There are $K$ source destination pairs $\{p_k:(a_k,b_k),k=1:K\}$ with the required flow $d_k$ between the source and destination of $k\ths$ pair. A path $P_k$ between a source destination pair $k$ consists of a set of connected edges {\em i.e.}
\begin{align*}
P_k=\{e_1,e_2,\cdots e_n\}: \text{such that } e_1=(a_k,s_1),e_2=(s_2,s_3),\\\cdots 
e_n=(s_n,b_k), s_i\in \nodes\  \forall i\in[1,n].
\end{align*}
 Let $\mathcal{P}_k=\{P_k\}$ denote the set of all path between the $k\ths$ pair. Let $\pathall=\cup_{k=1:K}{\mathcal{P}_k}$. Let the flow in each path $P\in \pathall$ be denoted by $f_P$.  For any edge (link) $e\in\edges$, the total traffc flow in the link is equal to  the total flow in that link as contributed from all paths of all the pairs $\flow_e=\sum\limits_{P \ni e} \pathflow_P$. Each link $e\in E$ has a general latency function $l_e(\cdot)$ such that the cost incurred in that link is equal to \[c_e=l_e(\flow_e).\]
This latency function depends on the ink characteristics, for example, type of the link, its capacity, construction materials.  
We assume the traffic network is heterogeneous so that the required flow $d_k$ can consists of different proportion of different types of users, as defined in the following subsection. 

\subsection{User Classes and General HetGame}
We assume that there are $M$ types of users where different types differ in  their proportion of the total population and traffic path selection strategy. Let $\usertypeset$ denote all types of users. We assume that a type $m\in\usertypeset$ has proportion $\alpha_m$ of the total demand flow for each $k\ths$ pair $p_k$. The complete traffic assignment problem can be seen as a $M$ player game termed as HetGames where the set of $m$ type users can be seen as a single player $m$. Let the combined strategy of this player $m$ is $X_m=(X_{mP})_{P\in\pathall}$.  Given the strategy, the flow of type $m$ users in a path $P$  is given as $\alpha_m X_{mP}$. 
The utility function of the $m\ths$ player is denoted as $u_m(X_m,X_{!m}\alpha)$, where $X=\{X_m:m\in\usertypeset\}$ is the combined strategy and $X_{!m}=\{X_j:j\in\usertypeset, j\ne m\}$.  
Thus, the total flow in any link $e$ for the heterogeneous traffic network  is
\begin{align*}
\flow_e&=\sum_{P \ni e}\sum_{m \in \usertypeset}{\alpha_mX_{m,P}}.
\end{align*}
To clarify the type of users, we describe some of the interesting users types in the following list:
\begin{enumerate}
\item Socialist:
Socialist users aim to minimize the total cost of the network which is given as
\begin{align*}
C(X_\socialist)=-\sum_{e\in\edges} {\flow_e l_e(\flow_e)}.
\end{align*}
In a traffic assignment problem with socialists users only, the optimal flow is given as the argmax of $C$ \cite{Nikolova2014}. Therefore, we can say that the utility function of the socialist player is average cost of the network which is given as
$u_\socialist=-\sum_e {\flow_e l_e(\flow_e)}$ \cite{Nikolova2014}.
\item Anarchist: 
Anarchist users aim to minimize their own cost, and therefore chose the path with the least link cost.  For anarchist flow, the optimal flow is given as the nash equilibrium of the traffic assignment problem. In the absence of any other class of users, the optimal flow for anarchist is given as the solution of the following problem \cite{Wardrop1952}
\begin{align}
X_\anarchist^\mathrm{opt}=\arg\min\sum_{e\in P:X_{\anarchist P}>0} {\int_0^{\flow_e}l_e(v)\dd v}
\end{align}
Therefore anarchist player utility can be given as 
\[u_\anarchist=-\sum_{e\in P:X_{\anarchist P}>0} {l_e(\flow_e)}.\]
\item Proportionally-Fair Socialist:  
The proportionally-fair strategy tries to minimize the total cost function while maintaining fairness---conceptually, fairness requires that not even a small fraction of users experience a particularly high cost. These users minimize the cost $u_\pfsocialist=-\sum_{e} {\flow_e \expS{{l_e(\flow_e)}}}$.
\end{enumerate}

\subsection{$\alpha$-Anarchy HetGame}
In the current paper, we will consider a particular HetGame with two classes of users: socialists with $\alpha_\socialist=(1-\alpha)$ proportion and anarchists with $\alpha_\anarchist=\alpha$ proportion. We term this game as $\alpha$-anarchy HetGame. For each source and destination pair $p_k$, the required anarchist flow is $\alpha d_k$ and socialist flow is $(1-\alpha)d_k$. 

Let $X_\socialist=x$ be the socialist strategy and $X_\anarchist =y$ be the anarchist strategy. For any pair $p_k$, the sum of the socialist flows in all the paths is $\sum\limits_{P_k\in\mathcal{P}_k}{(1-\alpha)x_{P_k}}$, which is required to be equal to the total demand $(1-\alpha)d_k$. This results in the following flow constraint: 
\begin{align*}
\mathcal{S}_\socialist: \sum_{P_k\in\mathcal{P}_k}{(1-\alpha)x_{P_k}} = (1-\alpha)d_k,\ \forall k  
\end{align*}
 for socialists. This constraint is equivalent to
\begin{align*}
\mathcal{S}_\socialist: \sum_{P_k\in\mathcal{P}_k}{x_{P_k}} =d_k,\ \forall k 
\end{align*}
 Similarly, for the anarchist, the flow constraint is given as 
\begin{align*}
\mathcal{S}_\anarchist:\sum_{P_k\in\mathcal{P}_k}{y_{P_k}}=d_k.
\end{align*}
Now, in any edge $e$, the total flow is equal to 
\begin{align}
\flow_e=\sum_{P \ni e, P\in\mathcal{P}}{(1-\alpha)x_{P}+\alpha y_{P}}.
\end{align} 

In the game where both classes co-exist, each class will try to optimize their own flow in presence of the flow of other class according to their own path selection strategy as described above. In the next section, we will develop a framework to derive the joint optimal flow for the users of the two classes.

\section{Joint Optimal Flow}
The $\alpha$-anarchy HetGame consists of simultaneous play between two types of users trying to minimize a different cost function. From a high level, the
anarchists will try to achieve nash equilibrium (NE) given the socialist flow and the socialists must find a socialist flow so that the responding NE strategy from the anarchists achieves the minimum social cost. In this section, we focus on developing a framework to derive the optimal flow. The following Theorem is particularly helpful in solving this two stage problem.

\begin{theorem}
\label{LemmaNEGen}
Given socialist strategy $x$ in any general $\alpha$-anarchy HetGame, the NE of anarchist users is given as the solution of:
\begin{align*}
y^{*}&=\arg\min_y \sum_{e}\int_0^{y_e}l_e(\sum_{P \ni e}{(1-\alpha)x_{P}+\alpha z})\dd z\\
&\text{such that } y_e=\sum_{P \ni e}{y_{P}}, \sum_{P_k \in \mathcal{P}_k}{y_{P_k}}=d_k.
\end{align*}
\end{theorem}

\begin{proof}
See Appendix \ref{proofforlemma}.
\end{proof}

Now, the socialist users (or the player) must choose a strategy $x^*$ such that the total cost of the network $C=\sum_{e}\flow_el_e(\flow_e)$ is minimized.
Therefore the optimal strategy is given as the solution $(x^*,y^*)$ of the optimization problem $\mathcal{S}$ which is simultaneous solution of the two following sub-problems $\mathcal{S}_1,\mathcal{S}_2$:
\begin{align*}
\mathcal{S}_1:\ \ \ \ x^*=&\arg\min_x \sum_{e}f_el_e(f_e)\\
\text{such that }& f_e=\sum_{P \ni e}{(1-\alpha)x_{P} + \alpha {y}^{*}_{P}(x)}, \sum_{P_k}{x_{P_k}}=d_k\\
\mathcal{S}_2:y^*(x)=&\arg\min_y \sum_{e}l'_e\left(y_e,x\right)\\
\text{such that }& y_e=\sum_{P \ni e}{y_{P}}, \sum_{P_k}{y_{P_k}}=d_k.
\end{align*}
where $l'_e(y_e,x)$ is the modified link cost function for  anarchist and given as
\begin{align}
l'_e(y_e,x)&=\int_0^{y_e}l_e\left(\sum_{P \ni e}{(1-\alpha)x_{P}+\alpha z}\right)\dd z.
\end{align}

The link cost function $l_e(\cdot)$  is generally taken as convex ({e.g.}, of the form $l_e(x_e)=b_e+a_e(x_e/c_e)^{d}$, $d\ge1$).   The above optimization is a convex problem for convex link cost functions and can be solved using the following alternative minimization \ref{algo_1} $\mathcal{S}_1$ given $y^*$ and $\mathcal{S}_2$ given $x^*$:
\begin{algorithm}
\caption{Alternating Minimization:}\label{algo_1}
\begin{algorithmic}
\STATE{ Initialize to $x^*_0,y^*_0$, $i=0$.}
\STATE{Solve optimization $\mathcal{S}_1$ to compute $x^*_1,y^*_0$. }
\STATE{ Solve optimization $\mathcal{S}_2$ to compute $x^*_1,y^*_1$. }
\WHILE{Change in solution is greater than tolerance}
\STATE{At step $i$, \begin{align*}
x^*_{i+1},y^*_i =\mathcal{S}_1(x^*_i,y^*_i)\\
x^*_{i+1},y^*_{i+1} =\mathcal{S}_2(x^*_{i+1},y^*_i)
\end{align*}}
\STATE{$i\rightarrow i+1$.}
\ENDWHILE{}

\end{algorithmic}
\end{algorithm}

\section{Analysis for Networks with Linear Latencies}
In this section, we will analyze some special cases and derive optimal flow for these special cases.  To compare the equilibrium/optimum performances, we define the following two terms which help in characterizing impact of a strategy.

\noindent \textbf{Price of $\alpha$-anarchy:} Price of $\alpha$ anarchy is defined as relative increase in the average cost due to presence of $\alpha$ proportion of anarchists {\em i.e.}
\begin{align*}
P_\mathrm{A}&=\frac{\text{Total cost with $\alpha$ anarchists and $(1-\alpha)$ socialists}}
{\text{Total cost of system with no anarchist}} 
\end{align*}
\textbf{Price of good behavior $P_G$:} Second important metric to understand the social implications of social strategy is price of good behavior which indicates the penalty a person may pay being a follower of the central directive. It is defined as the relative cost of following central directive compared to that when being selfish {\em i.e.} 
\begin{align*}
P_\mathrm{G}&=\frac{\text{Average cost of a socialist user}}
{\text{Average cost of an anarchist user}}.
\end{align*}

\subsection{Network with Linear Latencies}
Let us consider a traffic network with linear link cost functions {\em i.e.} $l_e(x_e)=a_ex_e+b_e$. Here $b_e$ is the free flow time and $a_e$ is congestion dependency parameter and both depend on link type. For example, freeways have high $b_e$ and low $a_e$ while city streets have higher $a_e$ and small $b_e$. For this case, the modified link cost function are given as
\begin{align*}
l'_e(y_e,x)&=\int_0^{y_e}\left(
a_e\left(
\sum_{P \ni e}{(1-\alpha)x_{P}}+\alpha z\right)+b_e\right)\dd z\\
&=\frac{a_e\alpha}{2}
 y_e^2+\left(b_e-{a_e}(1-\alpha)\sum_{P \ni e}{x_{P}}\right)y_e.
\end{align*}

So, the optimization problem $\mathcal{S}_1$ can be written as the following convex optimization problem:
\begin{align*}
x^*=&\arg\min_x \sum_{e}a_ef_e^2+b_ef_e\\
\text{such that }& f_e=\sum_{P \ni e}{(1-\alpha)x_{P} + \alpha {y}^{*}_{P}(x)}, \sum_{P_k}{x_{P_k}}=d_k, \ \forall k
\end{align*}
where $y^*(x)$ is given by 
\begin{align*}
y^*(x)=&\arg\min_y \sum_{e}\frac{a_e\alpha}{2}
 y_e^2+\left(b_e-{a_e}(1-\alpha)\sum_{P \ni e}{x_{P}}\right)y_e\\
\text{such that } &y_e=\sum_{P \ni e}{y_{P}}, \sum_{P_k}{y_{P_k}}=d_k.
\end{align*}

 \begin{figure}
 	\centering
     \includegraphics[width=.4\textwidth, trim=100 210 380 100, clip]{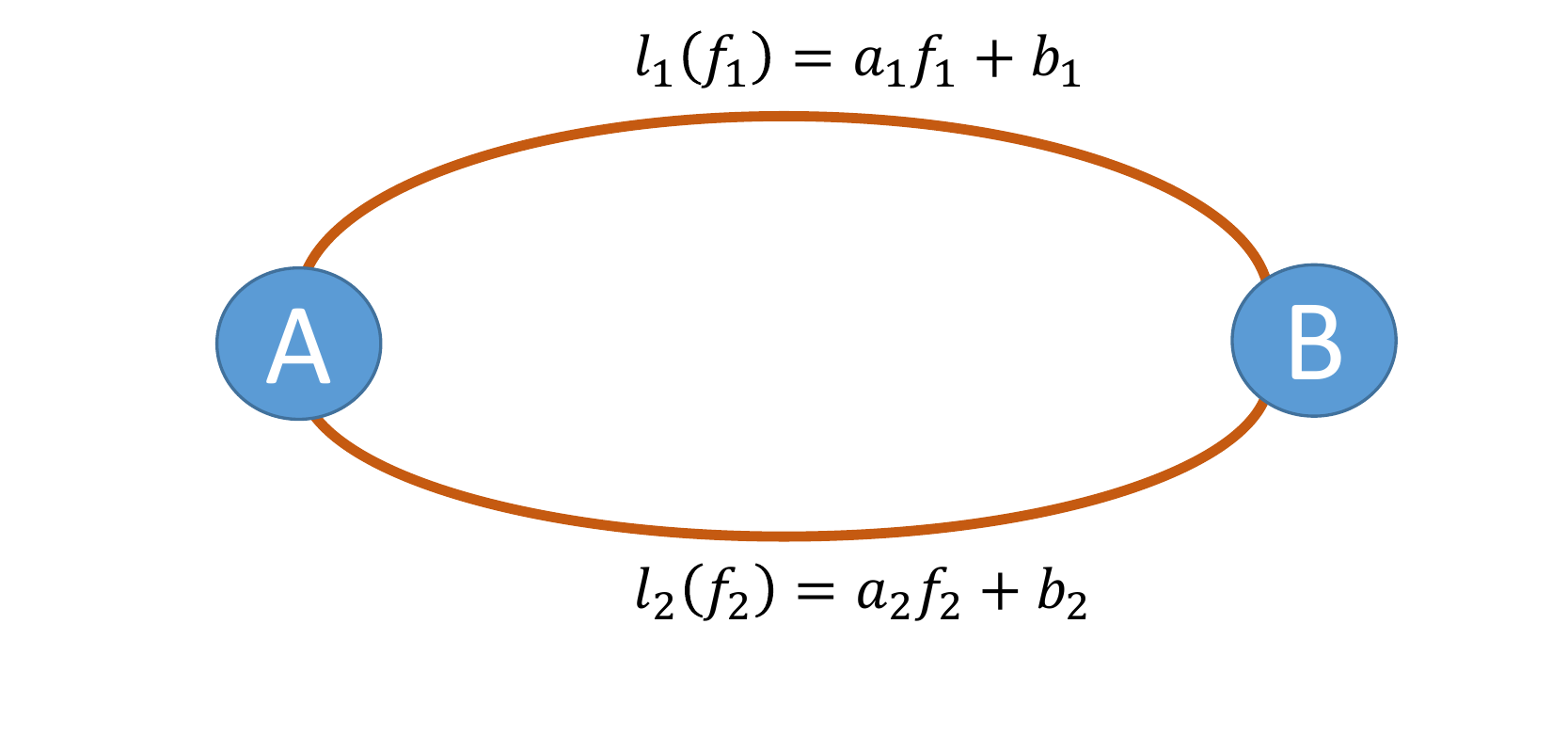}
     \caption{A illustration showing the two node two link network. The required demand flow is 1 between node A and B}
     \label{fig:twolink}
 \end{figure}

\subsection{General Two-link Linear Network}
As a special case of previous subsection, we will consider a general two node two link network (See Fig \ref{fig:twolink}) with linear latency. There are two possible paths  in the network and unit demand flow between the two nodes. The link $i(i=1,2)$ has latency $a_i\flow_i+b_i$ where $\flow_i$ is the flow in that link. Without loss of generality, let us assume that $b_1>b_2$. It can be observed easily that, in the absence of socialists (\ie $\alpha=1$), the equilibrium flow is given by NE \cite{Wardrop1952} as $f_1=1-A,\  f_2=A$ where $\displaystyle A=\frac{b_1-b_2+a_1}{(a_1+a_2)}$. We will assume that $0\le A \le 1$. Also, in the absence of the anarchist traffic (\ie $\alpha=0$), the social optimal solution is given as $f_1=1-f_{2\mathrm{opt}}, \ f_2=f_{2\mathrm{opt}}$ where $\displaystyle f_{2\mathrm{opt}}=\frac{a_1+(b_1-b_2)/2}{a_1+a_2}$.

Now for general $\alpha-$anarchy HetGame, let the anarchist strategy be $(y_1,y_2)$ and the socialist strategy be $(x_1,x_2)$. Using Theorem \ref{LemmaNEGen}, we can compute the NE of the anarchist users as solution of $\mathcal{S}_2$ given $x_2$ as
\begin{align}
y_2(x_2)&=
\begin{cases}
1 &\text{ if } R_1: \ x_2\le \frac{A}{1-\alpha}-\frac{\alpha}{1-\alpha}\\
\frac{A}{\alpha}-\frac{1-\alpha}{\alpha}x_2&\text{ if } R_2: \ \frac{A}{1-\alpha}\ge x_2\ge \frac{A}{1-\alpha}-\frac{\alpha}{1-\alpha}\\
0 & \text{ if } R_3: \ x_2\ge\frac{A}{1-\alpha}
\end{cases}
\label{Eq:y2givenx2}
\end{align}
where $\displaystyle A=\frac{b_1-b_2+a_1}{(a_1+a_2)}$. The above solution indicates that socialist can in fact indirectly force anarchist to chose an arbitrary strategy via a well designed socialist flow. 
 \begin{figure}[ht!]
	\centering
	\includegraphics[width=\figwidth]{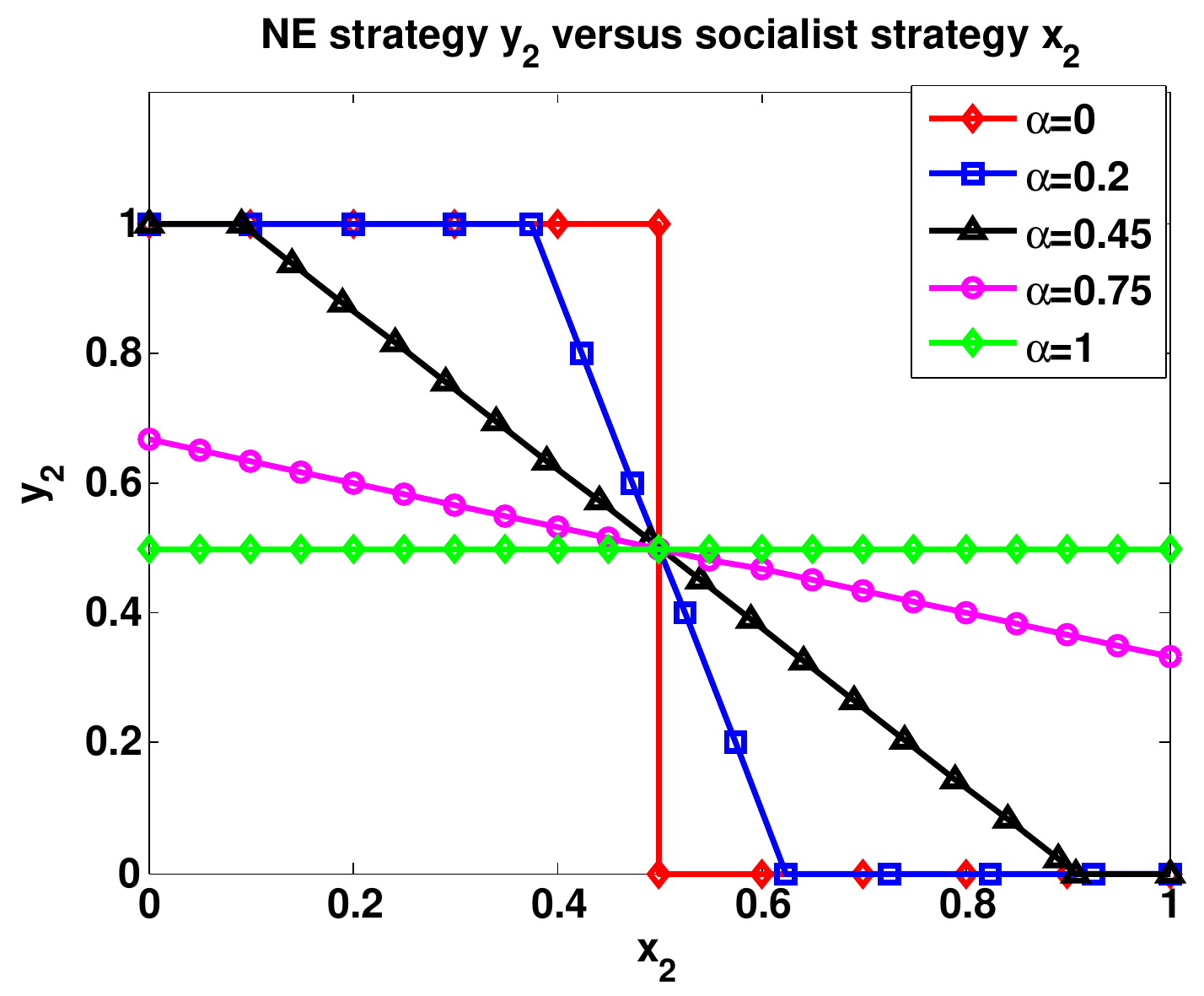}
	\caption{ The NE strategy flow $y_2$ in the second path versus socialist strategy flow $x_2$ of the second path for different values of $\alpha$ in  the considered two-node-two-link linear network with $a_1=0.3,b_1=1,a_2=0.7,b_2=0.8$. 
	 Socialists can in fact indirectly force anarchist to chose an arbitrary  strategy via a well designed socialist flow.  However, control can be limited for particular values of $	\alpha$, such as $ 	\alpha< A$ and $	\alpha>1-A$. }
	\label{fig:figy2}
\end{figure}

For better understanding, we will take a particular instance of the  above mentioned network with the following parameters: $a_1=0.3,b_1=1,a_2=0.7,b_2=0.8$.  In this case, the anarchist equilibrium is $(1-A=0.5,A=0.5)$ and optimal social flow is ($0.6,f_{2\mathrm{opt}}=.4$). Fig \ref{fig:figy2} shows the optimal NE for three regions. For $\alpha=0.2$, 
 $y_2$ shifts from  the value 1 in $R_1=(0,0.375)$ to the value 0 in $R_3=(0.625,1)$. This indicates that by diverting 20\% socialist traffic to the second path,  all anarchists can be forced to take the second path, while by diverting 70\% socialists to the second path, all anarchists can be forced to take the first path. 
It is possible that not all of the above regions exist for particular values of $\alpha$ which can restrict the fraction of anarchists which can be forced or affected by the central planner. For example, for $\alpha=0.75$, only $A/\alpha=66.67\%$ anarchists at max can be forced to take the second path. It can be shown that $R_1$ doesn't exist for $	\alpha<A$ and $R_2$ doesn't exist for $	\alpha>1-A$. Also note that as $\alpha$ increases, the impact of $x_2$ on $y_2$ decreases as evident from the slope in region $R_2$.

Now, given the NE strategy $y=(y_1,y_2)$, the socialists (or a central planner such as city government) will design the flow for socialist such that total cost $C$ is minimized over all the three regions.

1. Region $R_1$: the optimal solution in this region is $x_{2\mathrm{opt}}=(f_{2	\mathrm{opt}}-\alpha)/(1-	\alpha)$. Since $f_{2	\mathrm{opt}}<A$, $x_{2\mathrm{opt}}<(A-\alpha)/(1-\alpha)$ \ie contained in $R_1$. 

2. Region $R_2$: the total flow in second path $f_2$ is always equal to $A$ and total cost is  the same for all values of $x_2$. Hence no minima exists in this region.

3. Region $R_3$: the optimal solution in this region is $x_{2\mathrm{opt}}=(f_{2	\mathrm{opt}})/(1-	\alpha)$. Since $f_{2	\mathrm{opt}}<A$, $x_{2\mathrm{opt}}<(A)/(1-\alpha)$ \ie it falls outside the region $R_3$. Hence no minima exists in this region.

Now, based on the above discussion, we can now state the following result.
\newcommand{\ftopt}{f_{2\mathrm{opt}}}
\begin{theorem}
For the above-mentioned two link network with $\alpha$ anarchy, the following statements hold
\begin{enumerate}
	\item When $\alpha\le f_{2\mathrm{opt}}$, the optimal strategy for socialist and anarchist is 	
	\begin{align*}
	(x_1,x_2)&=\left(
	\frac{
		1-f_{2\mathrm{opt}}
	}
	{1-\alpha}
	,\frac{\ftopt-\alpha}{1-\alpha}
	\right)\\
	(y_1,y_2)&=(0,1).
	\end{align*}
	In this case, the total flow in two paths is  the social optimum flow as the socialist are able to compensate for the anarchist flow and bring the system to the social optimum.
	\item When $f_{2\mathrm{opt}}<\alpha\le A$, the optimal strategy for socialist and anarchist is 	
	\begin{align*}
	(x_1,x_2)&=(1,0)\\
	(y_1,y_2)&=(0,1).
	\end{align*}
	In this case, the socialist cannot compensate for the anarchist traffic. All anarchists take the second link and all socialists take the first link. The total flow in the network  is $(f_1,f_2)=(1-\alpha,\alpha)$.
	\item When $\alpha>A$, all strategies are optimum. The total flow is constant at $A$ and the two links offer the same cost of travel.
\end{enumerate} 
\end{theorem}

 Fig. \ref{spex} shows the cost versus $x_2$ for different values of $\alpha$. It can be observed that as $\alpha$ increase, the optimum value of $x_2$ decreases from $\ftopt$ until it reaches 0 at $\alpha=\ftopt$. 
 
 \begin{figure}[ht!]
 	\centering
 	\includegraphics[width=\figwidth,trim=0 0 0 20,clip]{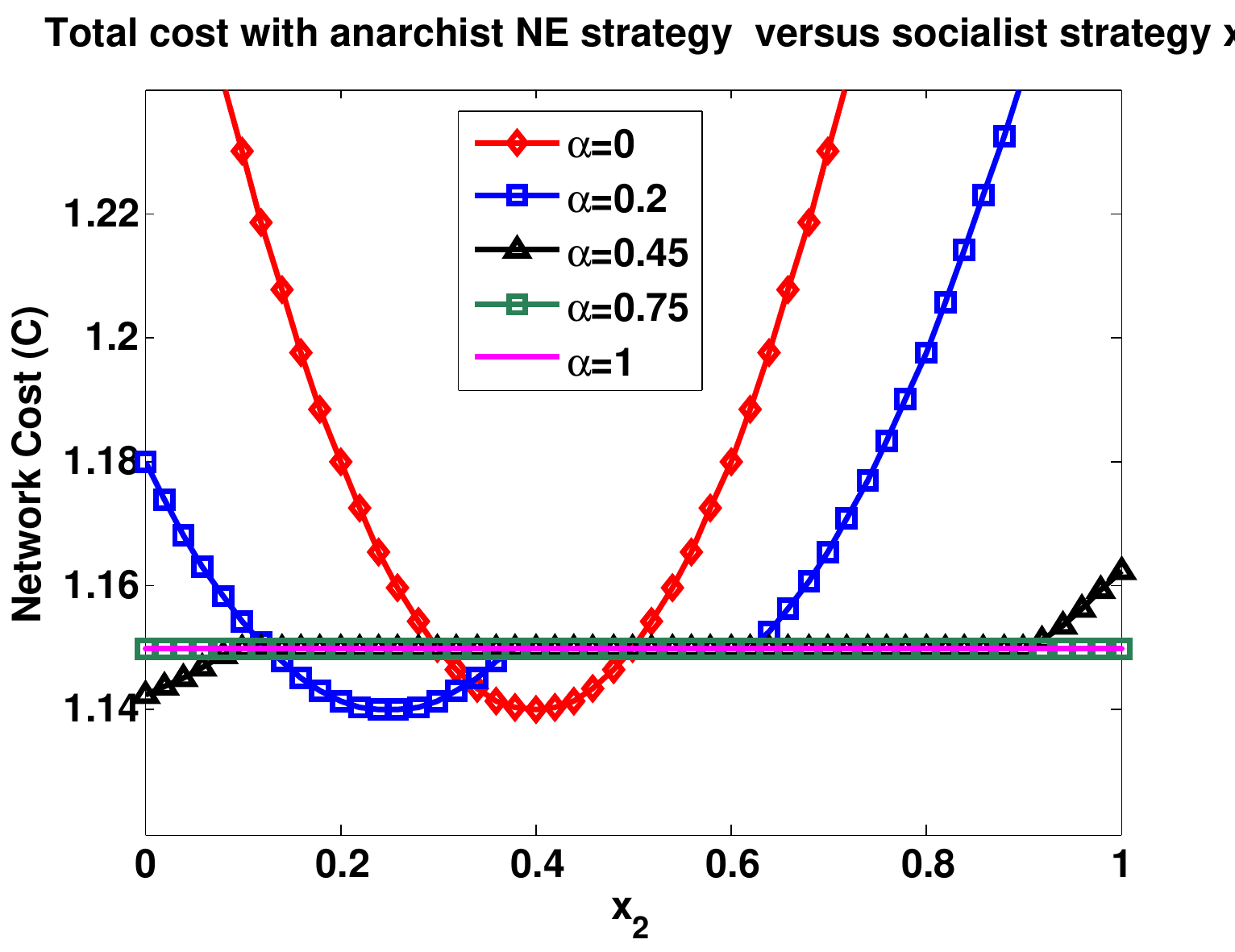}
 	\caption{Total cost of the network with NE strategy $y^*(x_2)$ versus socialist strategy $x_2$ for different value of anarchy $\alpha$ in a general two link HetGame with $a_1=0.3,a_2=0.7,b_1=1,b_2=0.8$. The cost is minimized in region $R_1$.}
 	\label{spex}
 \end{figure}
 \begin{figure}[ht!]
	\centering
	\includegraphics[width=\figwidth]{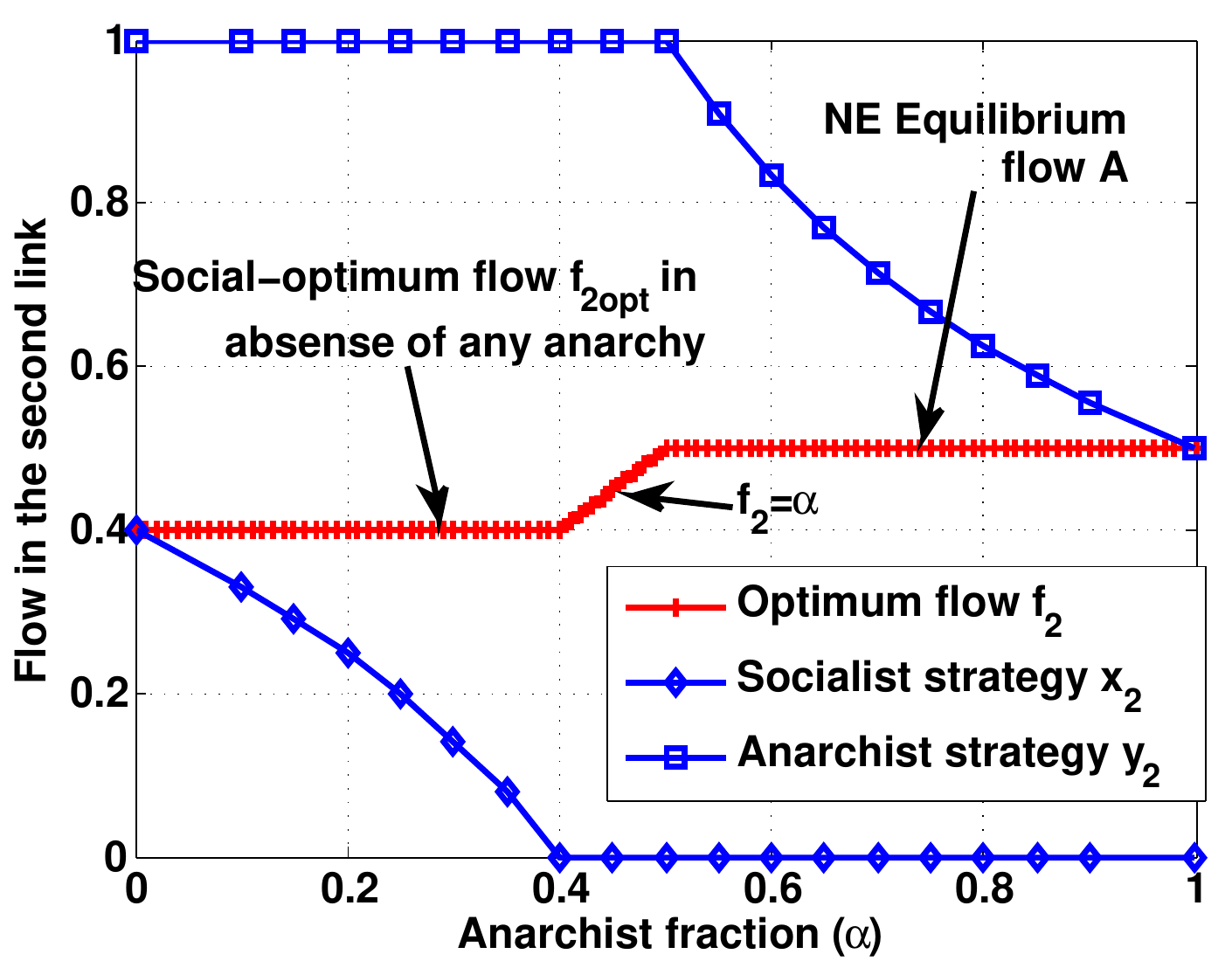}
	\caption{Equilibrium flow in a linear two link HetGame with $a_1=0.3,a_2=0.7,b_1=1,b_2=0.8$. Up to anarchy limited to $\ftopt$, the central planner is able to keep social equilibrium. After this, flow in second path increases until system reaches the pure anarchy (Nash equilibrium) state.}
	\label{fig:flowvsalpha}
\end{figure}
\begin{figure}[ht!]
	\centering
	\includegraphics[width=\figwidth]{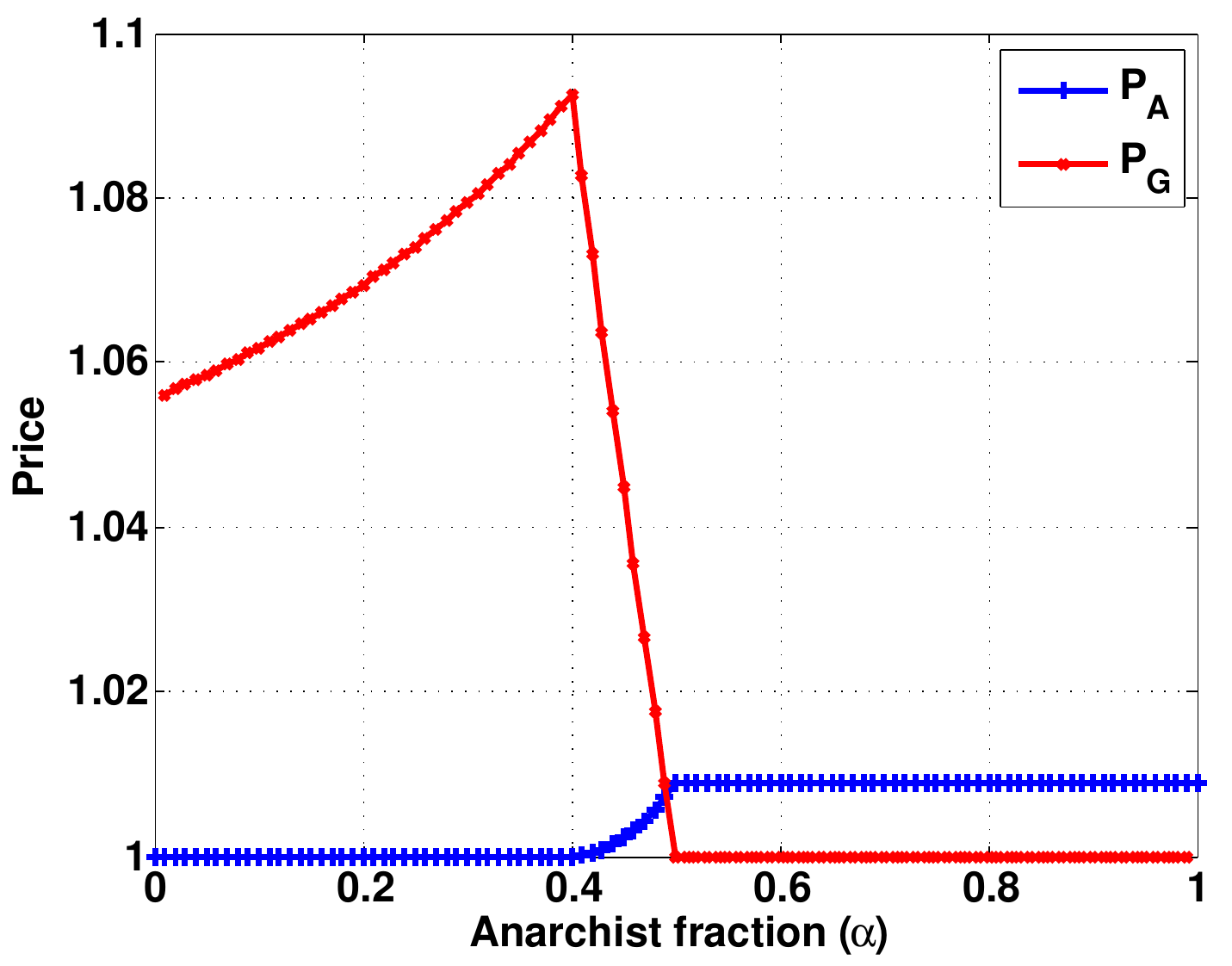}
	\caption{Price of $\alpha$ anarchy and price of good behavior with varying $\alpha$ for a linear two link HetGame with $a_1=0.3,b_1=1,a_2=0.7, b_2=0.8$. }
	\label{fig:prices}
\end{figure}

Fig. \ref{fig:flowvsalpha} shows the equilibrium flow in the network for different values of $\alpha$.
Fig. \ref{fig:prices} shows the variation of price of anarchy and good behavior with $\alpha$. It  shows that increasing anarchy will hurt anarchists also as evident from the increase in the price of anarchy. It can be seen that with increasing fraction of anarchist, price of being a good citizen increases, but after a threshold, it starts decreasing and eventually becomes equal to 1 where all users  start seeing the NE cost in both paths.

\begin{corollary}

Consider a simple two link  network with $a_1=0,b_1=1,a_2=1$ and $b_2=0$. In this case, link cost functions are $l_1(x)=1$ and $l_2(x)=x$. For this case, price of anarchy and good behavior are given as
\begin{align*}
P_\mathrm{A}&=\begin{cases}
1 & \text{ if } \alpha \le 1/2\\
\frac43(1-\alpha+\alpha^2) & \text{ if } \alpha > 1/2
\end{cases}\\
P_\mathrm{G}&=\begin{cases}
\frac{3-2\alpha}{2(1-\alpha)} \hspace{.45in}\ & \text{ if } \alpha \le 1/2\\
\frac1\alpha & \text{ if } \alpha > 1/2
\end{cases}.
\end{align*}
 	\end{corollary}

\section{Numerical Results}
To explore behavior of the HetGames laid out in previous sections, we studied a model of a real-world transportation network.  One publicly-available dataset is a 24-node model of the road network in Sioux Falls, SD.  The model characterizes the latency on 76 links connecting 24 nodes and provides trip data in the form of 528 origin-destination pairs. Each latency is a polynomial of form $l_e(x_e) = d\left(1+b(\frac{x_e}{c})^a\right)$.  
As a network grows, it is infeasible to enumerate the set of all paths $\mathcal{P}_k$ connecting a given source-destination pair. Therefore, we have selected $\mathcal{P}_k$ to be the four shortest zero-user paths using Dijkstra's algorithm.  

\begin{figure}[ht!]
	\centering
	\includegraphics[width=\figwidth]{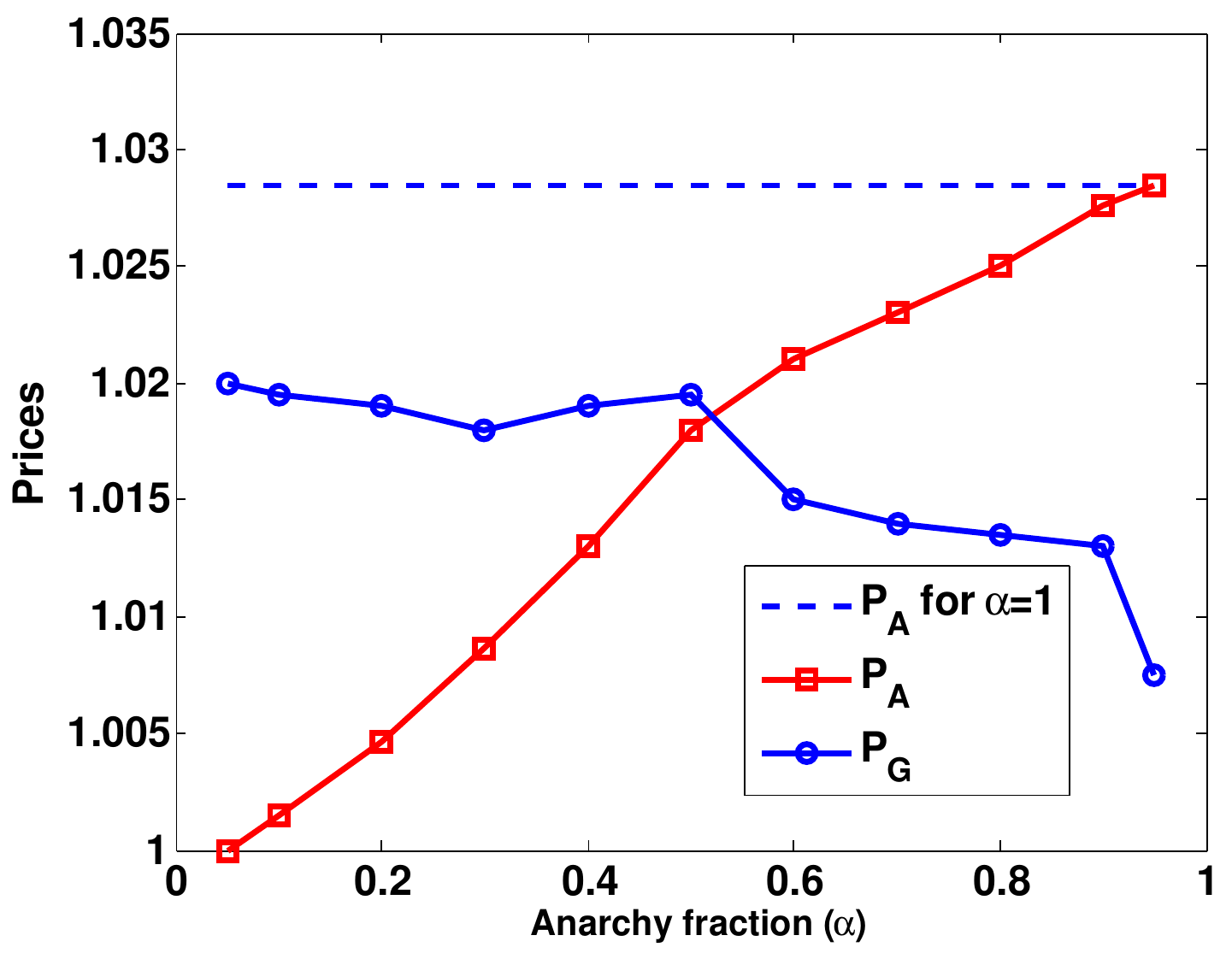}
	\caption{HetGame equilibrium for Sioux Falls network showing the price of $\alpha$-anarchy and Price of good behavior. Increase in the price of anarchy and decrease in the price of good behavior with increasing fraction of anarchists may motivate more people to be socialists and follower of central directives.}
	\label{fig:sioux}
\end{figure}

The proposed Algorithm-1 is used to compute the HetGame equilibrium flows.  In each iteration, the anarchist flow is computed holding the socialist flow constant; similarly, the socialist flow is computed holding the anarchist flow constant.  If the updated flows do not differ significantly from those in the previous iteration, the algorithm may terminate.  Fig. \ref{fig:sioux} shows the price of $\alpha$-anarchy and price of good behavior of the HetGame equilibrium for various $\alpha$ values. 

Traditional definition of price of anarchy (with $\alpha=1$) is also shown in the figure. It can be seen that price of good behavior decreases with increasing fraction of anarchist. It also shows that increasing anarchy will hurt anarchists also. Both of the above observation can lead to motivating more people to be follower of central directives.

\section{Conclusions and Future Work}
In this work, we derived a framework for general traffic network with heterogeneous users, including an $\alpha$-anarchy HetGame and studied it both analytically and via numerical simulation on road network models. We discussed how social flow can be used to affect the flow of anarchists in desired direction. The work has many possible extensions. For example, HetGame with more number of user classes (e.g., proportionally fair, random, and fixed path-followers) can be studied. It is also interesting to analyze these systems in the presence of random noise.  We can also consider the case where the socialists in the HetGame could modify their objective to be minimize of average socialist cost, neglecting the anarchist portion of the full social cost. Such strategy has potential to reduce the price of good behavior providing incentive to people who wish to adhere to central directives.
\appendices

\section{Proof for Theorem \ref{LemmaNEGen}}
\label{proofforlemma}
\begin{proof}
Given the socialist strategy $x$, the anarchist will decide their flow by computing the NE for their flow $\alpha y$. Given $x$, the latency they face in any edge $e$ is given by
$L_e(\sum_{P \ni e}{\alpha y_{P}})$ where $L(w)=l_e(\sum_{P \ni e}{(1-\alpha)x_{P}+\alpha w})$ for any strategy $y$. Let us assume $y'_e=\sum_{P \ni e}{\alpha y_{P}}$. With these new latency function, the NE of the flow following anarchist strategy will be equal to Wardrop equilibrium \cite{Wardrop1952} which is given as the solution of the following problem
\begin{align*}
&\min_y \sum_{e}\int_0^{y'_e}L_e(w)dw\\
&\text{such that } y'_e=\sum_{P \ni e}{\alpha y_{P}}, \sum_{P_k}{y_{P_k}}=d_k
\end{align*}
which is equivalent to
\begin{align*}
&\min \sum_{e}\int_0^{y'_e}l_e(\sum_{P \ni e}{(1-\alpha)x_{P}+w})dw\\
&\text{such that } y'_e=\sum_{P \ni e}{\alpha y_{P}}, \sum_{P_k}{y_{P_k}}=d_k
\end{align*}
which will give the desired result with substitution $y'_e=\alpha y_e$ and $z=\alpha w$.
\end{proof}

\end{document}